\documentclass[11pt,reqno]{amsart}

\usepackage{amsmath}
\usepackage{amssymb}
\usepackage{amsthm}
\usepackage{mathtools}
\usepackage{color}
\usepackage{graphicx}
\usepackage{psfrag}
\usepackage{marginnote}
\usepackage{lscape}

\newcommand{\geqs}{\geqslant}
\newcommand{\leqs}{\leqslant}
\newcommand{\Nb}{{\mathbb N}}
\newcommand{\Rb}{{\mathbb R}}

\newcommand{\xeq}{x^{\text{eq}}}
\newcommand{\Meq}{M^{\text{eq}}}
\newcommand{\ueq}{u^{\text{eq}}}
\newcommand{\veq}{v^{\text{eq}}}
\newcommand{\weq}{w^{\text{eq}}}
\newcommand{\Ueq}{U^{\text{eq}}}

\makeatletter
\renewcommand*\env@matrix[1][\arraystretch]{%
  \edef\arraystretch{#1}%
  \hskip -\arraycolsep
  \let\@ifnextchar\new@ifnextchar
  \array{*\c@MaxMatrixCols c}}
\makeatother

\newtheorem{theorem}{Theorem}[section]
\newtheorem{lemma}[theorem]{Lemma}
\newtheorem{prop}[theorem]{Proposition}

\newtheorem{conj}[theorem]{Conjecture}
\newtheorem{remark}[theorem]{Remark}

\theoremstyle{definition}
\theoremstyle{remark}

\begin{document}
	
\title[Dynamics of a silicosis model]{modelling silicosis: dynamics of a model with piecewise constant rate coefficients}

\author[P.R.S. Antunes]{Pedro R.S. Antunes} 
\address[P.R.S. Antunes]{Univ. Aberta, Dep. of Sciences and Technology,
  Rua da Escola Polit\'ecnica 141-7, P-1269-001 Lisboa, Portugal, and
  Univ. Lisboa, Faculty of Sciences, Group of Mathematical Physics,
Edif\'icio 6, Piso 1,  Campo Grande,
  P-1749-016 Lisboa, Portugal.}  \email{Pedro.Antunes@uab.pt}

\author[F.P. da Costa]{Fernando P. da Costa}
\address[F.P. da Costa]{Univ. Aberta, Dep. of Sciences and Technology,
  Rua da Escola Polit\'ecnica 141-7, P-1269-001 Lisboa, Portugal, and
  Univ. Lisboa, Instituto Superior T\'ecnico, Centre for Mathematical
  Analysis, Geometry and Dynamical Systems, Av. Rovisco Pais,
  P-1049-001 Lisboa, Portugal.}  \email{fcosta@uab.pt}

\author[J.T. Pinto]{Jo\~ao T. Pinto} 
\address[J.T. Pinto]{Univ. Lisboa, Instituto Superior T\'ecnico, Dep. of Mathematics
and Centre for Mathematical
  Analysis, Geometry and Dynamical Systems, Av. Rovisco Pais,
  P-1049-001 Lisboa, Portugal.}  \email{jpinto@tecnico.ulisboa.pt}

\author[R. Sasportes]{Rafael Sasportes} 
\address[R. Sasportes]{Univ. Aberta, Dep. of Sciences and Technology,
  Rua da Escola Polit\'ecnica 141-7, P-1269-001 Lisboa, Portugal, and
  Univ. Lisboa, Instituto Superior T\'ecnico, Centre for Mathematical
  Analysis, Geometry and Dynamical Systems, Av. Rovisco Pais,
  P-1049-001 Lisboa, Portugal.}  \email{rafael.sasportes@uab.pt}

\thanks{Research partially supported by Funda\c{c}\~ao para a Ci\^encia e a Tecnologia (Portugal) 
through project CAMGSD UID/04459/2020.}
\thanks{\emph{Corresponding author:} F.P. da Costa}

\date{September 2, 2021}



\begin{abstract}
We study the dynamics about equilibria of an infinite dimension coagulation-fragmentation-death model for
the silicosis disease mechanism introduced recently by da Costa, Drmota, and Grinfeld (2020) \cite{cdg}
in the case where the rate coefficients are piecewise constant.

\end{abstract}


\maketitle
\section{Introduction}\label{sec1}


Silicosis is an incurable, long-term lung disease, caused by breathing in dust that 
contains crystalline silica, which is commonly found in sand, rock, and mineral 
ores like quartz. Artificial stone containing high levels of silica can also become 
dangerous for workers manipulating it; see  \cite{arti} for a recent review.  

We give a brief description of the processes involved in the lungs immune system's 
response to the invasion by harmful silica dust particles.

When silica dust particles reach the lungs they trigger a response from the 
alveolar macrophages, either through chemotaxis or by chance encounters. 
The next step is the engulfment and removal of the pathogens and cell 
debris by the alveolar macrophages, this process is known as phagocytosis. 
In the lungs, three different populations of macrophages exist, including 
airway, alveolar, and interstitial macrophages. Alveolar macrophages are 
situated on the inner surface of the lung, and they account for 55\% of the 
lung immune cells.

There is a diverse set of pathologies associated with silica exposure, so it 
seems unlikely that there is a single common mechanism responsible for all 
of the possible diseases. The exact sequence of events (from silica inhalation 
to disease) is unknown, but it is generally accepted that the alveolar 
macrophage plays a relevant role. 
Upon contact, the alveolar macrophage will bind to the silica and begin to
 engulf the particle. If the alveolar macrophage survives the silica encounter,
  it will likely migrate out of the lungs to either the proximal lymph nodes or 
  through the mucosal-ciliary escalator and eventually out of the respiratory 
  tract. If the alveolar macrophage stays in the lung it will migrate to the 
  interstitial space and become an activated interstitial macrophage  that 
  could contribute directly to worsen the disease \cite{hamilton}.
 Although the reasons for the underlying mechanism are not clear, 
silica particles are toxic to the macrophages \cite{gilberti} and can lead to their 
death. 
If this happens while the macrophages are still in the lungs,  the silica particles 
are released back into the respiratory system.

\medskip

The probability for a given macrophage already containing $i$ particles of silica
 to engulf an additional particle typically decreases with $i$ and in the model in \cite{tran} a maximum load capacity of 
$n_\text{max}<\infty$
is assumed \textit{a priori}. In \cite{cdg} this restriction was not explicitly considered, being 
the existence of an effective upper bound of the silica particles' load of the macrophages left as a consequence of the 
assumptions upon the rate coefficients.
Because of the toxicity of silica particles to the macrophages referred to above, 
macrophages with a higher load of silica particles will die 
at a higher rate. Moreover, the ability of the 
macrophages to migrate through the mucociliary escalator is  impaired by an increase 
in their load of silica particles. It is the balance of these processes that leads to the
mathematical model in \cite{tran} and that we also consider here (and was already considered in \cite{cps}) 
with the changes introduced in \cite{cdg}.

\medskip

Let $M_i = M_i(t)$ be the concentration of macrophages which contain $i$
silica particles (we will refer to it as the $i$-th cohort) at time $t$,  $x=x(t)$ be the
concentration of silica particles, and $r$  the
rate of supply of new (with no silica particles) macrophages. Following the model considered in \cite{tran}, we obtain the
 equations for the mechanism described above:
\begin{align}
    \frac{dM_0}{dt} & = r - k_0 x M_0 - (p_0+q_0) M_0,\label{M0eq1}\\
    \frac{dM_i}{dt} & = k_{i-1} x M_{i-1} - k_i x M_i -(p_i+q_i) 
    M_i, \;\; i \geqslant 1, \label{Mieq2}
\end{align}
where  $k_i$ is the rate of phagocytosis of a silica particle by a macrophage already containing
$i$ particles, $p_i$, is the transfer rate of macrophages in
the $i$-th cohort to the mucociliary escalator, i.e. the rate of
their removal from the pulmonary alveoli together with their quartz load, and $q_i$ is the
rate of death of the macrophages in the $i$-th cohort which results in the release of the
quartz burden back into the lungs.  As stated above the model in \cite{cdg}, unlike the one in  \cite{tran},  
does not impose an upper limit on  the
number $i$ of quartz particles a macrophage can contain, the existence of such a
load capacity will be a consequence of the assumptions on the rate coefficients $k_i$ and $q_i$.

\medskip

The following governing equation for the evolution of the concentration of
silica particles in the system was considered
 in \cite{cdg}, under the assumption of an inhalation rate $\alpha$, 
valid under the same assumption about the validity of the
mass action law used to obtain the equations for the $M_i$:
\begin{equation}
  \frac{dx}{dt} = \alpha - x \sum_{i=0}^\infty  k_i M_i + \sum_{i=0}^\infty q_i i M_i.  \label{xeq3}
\end{equation}
The second term in the right-hand side models the decrease in the concentration of free silica particles due to their
ingestion by macrophages, and the third term represents their increase due to them being
released into the lungs when macrophages die.
A kinetic scheme of the processes modelled by the rate equations \eqref{M0eq1}--\eqref{xeq3} 
is  presented in Figure~\ref{fig1}, \cite{cps}.

%
%
%
\begin{figure}[!h]
	\includegraphics[scale=1.05]{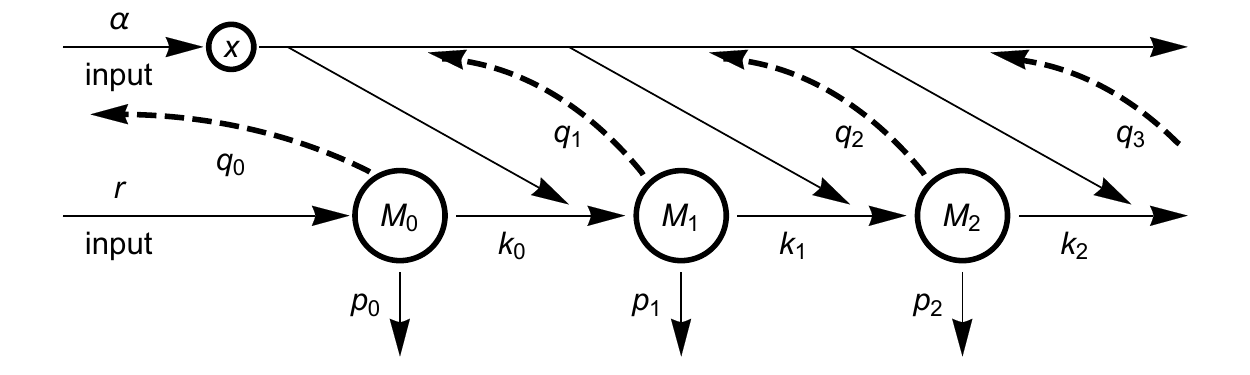}
	\caption{Reaction scheme of the model considered in this paper \cite{cps}. The input
		rates of quartz and of macrophages with no quartz
		particles are $\alpha$ and $r$, respectively. The concentration of free
		quartz particles and of macrophages containing $j$ quartz
		particles are represented by $x$ and $M_j$, respectively. 
		Macrophages $M_j$ can be destroyed, releasing $j$ quartz particles (dashed lines),
		or they can be removed by rising in the mucocilliary escalator (vertical downward
		lines), or they can ingest an additional quartz particle
		becoming an $M_{j+1}$ macrophage (horizontal rightward arrows).}\label{fig1}
\end{figure}
%
%
%

\medskip

For the functional setting in which to study \eqref{M0eq1}--\eqref{xeq3} we consider the set of
 elements \(y=(y_n)=(x,M_0,M_1,\dots)\)  of \(\mathbb{R}^\mathbb{N}\) defined by
\[ X=\{y=(y_n): \|y\|<\infty\}\]
where 
\[
\|y\|:=|x|+\sum_{i=0}^\infty (i+1)|M_i|=|x|+|((i+1)M_i)_{i=0,1,\dots}|_{\ell^1}.
\]
It is clear that $(X, \|\cdot\|)$ is a Banach space (and a subspace of $\ell^1$).
We say that \(y\geqslant 0\) if and only if \(y\in(\mathbb{R}_0^+)^\mathbb{N}\) and we denote the nonnegative cone of \(X\) by \(X_+:=\{y\in X :  y\geqslant 0\}.\)  

\medskip

From a biological point of view we are only interested in nonnegative solutions  of 
\eqref{M0eq1}--\eqref{xeq3}, $y(t)\geqs 0$ for all $t\geqs 0$. If $y(t)\in X_+$ then
the quantity
$\|y(t)\|$ represents the total amount of particles (macrophages cells and silica particles inside and outside the macrophages)
per unit volume at time $t$, and so working in $X_+$ corresponds to consider solutions of  \eqref{M0eq1}--\eqref{xeq3} with
finite amount of particles per unit volume. 

\medskip

Existence, uniqueness, continuous dependence and semigroup property of solutions to the Cauchy 
problem for the  infinite dimensional system of ordinary differential equations \eqref{M0eq1}--\eqref{xeq3} were studied
in \cite{cps}. Aspects of the structure of equilibria were analyzed in \cite{cdg}. In this paper we consider 
aspects of the long time behaviour of solutions for the system with
the following class of
piecewise constant coefficients introduced and studied in \cite[Section 3.1]{cdg}:
\begin{equation}
k_i \equiv k, \qquad
p_i = \begin{cases} 1 & \hbox{if $i \leq N,$}\\ 0 & \hbox{if $i \geq N+1,$}\end{cases}
\quad\hbox{and}\quad
q_i = \begin{cases} 0 & \hbox{if $i \leq N,$}\\ 1 & \hbox{if $i \geq N+1,$} \end{cases}\label{coef}
\end{equation}
for some fixed positive integer $N.$

\medskip

With these coefficients system \eqref{M0eq1}--\eqref{xeq3} becomes
\begin{align}
    \dfrac{dM_0}{dt} & ~= r - k x M_0 - M_0,\nonumber \\
    \dfrac{dM_i}{dt} & ~= k x M_{i-1} - k x M_i - M_i, \;\; i \geqslant 1,\label{syst2}  \\
    \dfrac{dx}{dt} & ~= \alpha - kx \sum_{i=0}^\infty   M_i + \sum_{i=N+1}^\infty i M_i, \nonumber
\end{align}

\noindent
and the structure of its equilibria is completely understood and was proved in
Propositions 1 and 2 of \cite{cdg}: 
\begin{prop}\label{prop:structure}
For all $N\in\Nb$ and $ k >0$, there exists a unique $\mu^*>0$ such that \eqref{syst2} has:
\begin{enumerate}
\item no equilibria if $\alpha/r > \mu^*$,\\
\item exactly one equilibrium if $\alpha/r = \mu^*$,\\
\item exactly two equilibria if $\alpha/r \in (0, \mu^*).$
\end{enumerate}
\end{prop}

\medskip

As was proved in \cite{cdg} and will be recalled below, each
equilibrium solution of  \eqref{syst2}, $(x^\text{eq}, M_0^\text{eq}, M_1^\text{eq}, M_2^\text{eq}, \ldots),$
can be identified by its $x$ component, and Proposition~\ref{prop:structure}  can be graphically 
depicted by the bifurcation diagram presented in Figure~\ref{figbif}.

%
%
%
\begin{figure}[!h]
	\includegraphics[scale=1.2]{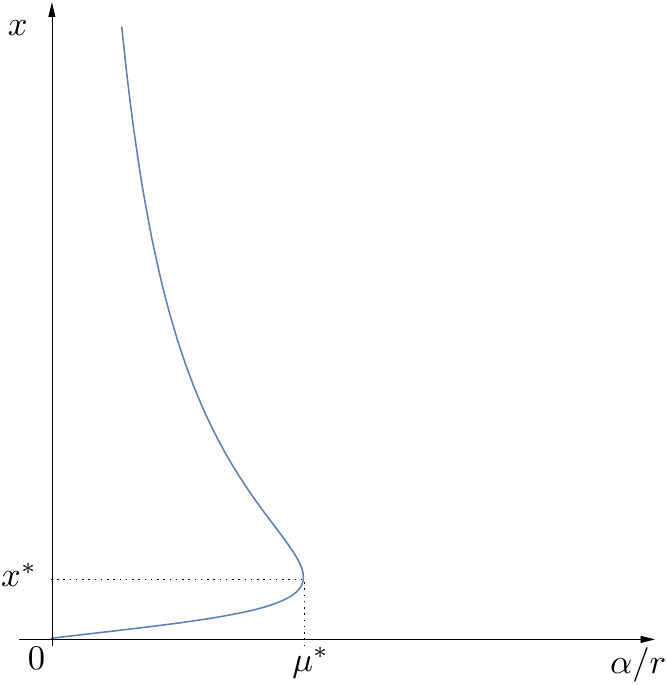}
	\caption{Bifurcation diagram of equilibria of \eqref{syst2}.}\label{figbif}
\end{figure}
%
%
%

\medskip

In this paper we study the local stability of the equilibria of the silicosis system \eqref{syst2}. The paper is organized as follows:
 
 In Section \ref{sec2} we recall some results obtained in \cite{cdg} about the time independent solutions of the silicosis system.
In particular we present a bifurcation equation, \eqref{bifeq}, whose solutions give (the component $x$ of) the equilibria of
\eqref{syst2} and point out properties of the bifurcation function that will be crucially important in the remaining of the paper.

\medskip

Informally, the result about the number of equilibria (already presented above in Proposition~\ref{prop:structure} and Figure~\ref{figbif})
states that if the balance between the input rates of silica, $\alpha$, and of macrophages, $r$, is such that, in some sense,
the silica input does not dominate, then \eqref{syst2} will have two equilibria and, from an heuristic viewpoint, we expect 
the equilibria with a smaller load of silica dust $x$ to be stable while the other is expected to be unstable. It is the goal of this
paper to make this argument rigorous and to prove this intuition.

\medskip

In Section \ref{sec3} we show that by introducing three bulk variables $u, v$ and $w$ defined by \eqref{uvw}  
system \eqref{syst2} is transformed into an infinite dimensional system \eqref{syst3} for the variables
$(x, u, v, w, M_0, \ldots)$ for which the equations for the variables $(x, u, \ldots, M_{N})$
consist of a closed $(N+5)$-dimension system of ordinary differential equations. It is this decoupling that allows
the study of the stability of the equilibria of \eqref{syst2} to be done by  first
obtaining appropriate results about the spectra of the linearizations
around the equilibria of this finite dimensional system, which is done in the remainder of section~\ref{sec3}. 

\medskip

In Section \ref{sec4} we study the local dynamics about the equilibria of the  full silicosis system \eqref{syst2} using the results
about the stability properties of the equilibria of the reduced $(N+5)$-dimensional system obtained in the previous section.
In particular we prove that our intuition was correct: the equilibrium of \eqref{syst2} with a lower load of silica dust $x^{\text{eq}}$ 
is locally exponentially asymptotically stable 
in the strong topology of $X$, whereas the equilibrium with a higher silica load $x^{\text{eq}}$ is unstable.

\medskip

In Section \ref{sec5} 
we present some  numerical evidence illustrating the spectra of the jacobian matrices of the linearizations of the $(N+5)$-dimensional system about the equilibria. These numerical experiments support the conjecture in section~\ref{sec3} about the dimension of the unstable manifold of the unstable equilibrium of the $(N+5)$-dimensional system, and also suggest that, besides those
properties proved in section~\ref{sec3}, which are relevant to our goal in this paper, the spectra has some other
features that could be interesting to explore in the future.

%
%
\section{Preliminaries: the equilibria}\label{sec2}
In this section we recall some of the results obtained in \cite{cdg} for the time independent solutions of the silicosis system
\eqref{syst2}. In that article the authors solve  equations \eqref{M0eq1}--\eqref{Mieq2} with  all the time derivatives equal to zero, thus obtaining the following expressions for the $M_i^\text{eq}$ variables corresponding to the equilibrium solutions, for general coefficients, $k_i>0$, $p_i\geqslant 0$, $q_i\geqslant 0$, in terms of the variable $x^\text{eq}$:
\begin{equation}\label{Mieq}
M_i^\text{eq}={\displaystyle \frac{r(x^{\text{eq}})^i}{k_i{\displaystyle\prod_{j=0}^i}(x^{\text{eq}}+d_j)}},\qquad i\geqslant 0,
\end{equation}
where $d_j=(p_j+q_j)/k_j.$ For our particular choice of the coefficients, that is, for system \eqref{syst2}, $d_j=1/k$,  and from \eqref{Mieq} they are easily obtained,
\begin{align}
\sum_{i=0}^\infty k_iM_i^{\text{eq}}&=k\sum_{i=0}^\infty M_i^{\text{eq}}=rk, \label{SumkM}\\
\sum_{i=0}^\infty iq_iM_i^{\text{eq}}&=\sum_{i=N+1}^\infty iM_i^{\text{eq}}=
r\big(kx+(N+1)\big)\left(\frac{x^{\text{eq}}}{x^{\text{eq}}+1/k}\right)^{N+1}. \label{SumiqM}
\end{align}
Plugging \eqref{SumkM} and \eqref{SumiqM} into the time independent version of the equation for the quartz concentration $x$ in system \eqref{syst2}, they obtain the bifurcation equation, 
\begin{equation}\label{bifeq}
\frac{\alpha}{r}-\mathcal{F}_{N,k}(x^{\text{eq}})=0,
\end{equation}
where, for all positive $x$,
\[
\mathcal{F}_{N,k}(x):=kx\left(1-\left(\frac{x}{x+1/k}\right)^{N+1}\right)-(N+1)\left(\frac{x}{x+1/k}\right)^{N+1}.
\]

Proposition \ref{prop:structure} in the previous section follows from the analysis of equation \eqref{bifeq}
that we briefly recall now: introducing the  variable $y=\frac{x^{\text{eq}}}{x^{\text{eq}}+1/k}$,
 and defining the function $\widetilde{\mathcal{F}}_{N}$ by
 \[
 \widetilde{\mathcal{F}}_N(y):=\frac{y}{1-y}\left(1-(N+1)y^{N}+Ny^{N+1}\right),
 \]
equation \eqref{bifeq} can be written as,
 \begin{equation}\label{bifeqy}
 \frac{\alpha}{r}-\widetilde{\mathcal{F}}_N(y)=0.
 \end{equation}
 Observe the independence of this bifurcation equation relatively to $k$: this coefficient only dictates how the variables $x^{\text{eq}}$ and $y$ are interrelated. Of relevance to our work are the arguments used in the proof of Proposition 2 in \cite{cdg} based on the study of the derivative,
 \begin{align}\label{Fprime}
 \widetilde{\mathcal{F}}_N'(y)=\frac{p_N(y)}{(1-y)^2},
 \end{align}
 where, 
 \begin{align}\label{pNy}
 p_N(y):=1-(N+1)^2y^N+N(2N+3)y^{N+1}-N(N+1)y^{N+2}.
 \end{align}
 The authors prove that $p_N$ is strictly decreasing in $\left(0,\frac{N+1}{N+2}\right)$
 and strictly increasing in $\left(\frac{N+1}{N+2},1\right)$. Since $p_N(1)=0,$ then
 $p_N\left(\frac{N+1}{N+2}\right)<0$ and, by the fact that $p_N(0)>0,$ it can be concluded that
 there is one and only one critical point $y^*$ of $\widetilde{\mathcal{F}}_N$ in $(0,1),$ and furthermore it satisfies, 
$y^*<\frac{N+1}{N+2}.$ This corresponds to the critical point $x^*$ of $\mathcal{F}_{k,N}$ that, 
together with \eqref{bifeq}, gives the bifurcation point $(\mu^*,x^*),$ displayed in the bifurcation diagram of figure \ref{figbif}.

%
%
%
%
\section{A finite dimensional reduced system}\label{sec3}
We start by showing that the dynamics of system \eqref{syst2} is dictated by a finite dimensional ODE. The characterization of the stability properties of our silicosis system will then be based on the study of this ODE.

Let us introduce the following three new variables:
\begin{align}
u:=\sum_{i=0}^\infty M_i,\qquad v:=\sum_{i=N+1}^\infty iM_i,\qquad w:=\sum_{i=N}^\infty M_i.\label{uvw}
\end{align}
By \cite[Corollary~5.3]{cps} the series in \eqref{uvw} are uniformly convergent. 
For any positive integer $m,$ we get, from \eqref{syst2},
\begin{align}
\sum_{i=0}^m \dot{M}_i &= r-kxM_0-M_0+kx\sum_{i=1}^m(M_{i-1}-M_i)-\sum_{i=1}^m M_i \nonumber \\
&= r - mM_m - \sum_{i=0}^mM_i\nonumber
\end{align}
and, if $m>N,$
\begin{align}
\sum_{i=N+1}^m i\dot{M}_i &= kx\sum_{i=N+1}^m i(M_{i-1}-M_i) - \sum_{i=N+1}^miM_i \nonumber \\
&=kx(NM_N-mM_m) + kx\sum_{i=N}^{m-1}M_i - \sum_{i=N+1}^miM_i.\nonumber
\end{align}
Hence, by the uniform convergence as $m\to \infty$ of the right-hand sides  of these equalities 
we conclude the left-hand sides are also uniformly convergent and since \cite[Proposition 6.1]{cps} 
ensures that  $M_i\in C^1([0,\infty))$, we conclude that $u, v$ and $w$ in \eqref{uvw} are $C^1$
functions and their derivative can be computed differentiating the series term-by-term:
\begin{align}
\dot{u}&=\sum_{i=0}^\infty \dot{M}_i =r-kxM_0-M_0+kx\sum_{i=1}^\infty(M_{i-1}-M_i)-\sum_{i=1}^\infty M_i\nonumber\\
&=r-u.\label{udot}\\
\dot{v}&=\sum_{i=N+1}^\infty i\dot{M}_i = kx\sum_{i=N+1}^\infty (iM_{i-1}-iM_i)-\sum_{i=N+1}^\infty iM_i\nonumber\\
&=kx\left(\sum_{i=N+1}^\infty \big((i-1)M_{i-1}-iM_i\big)+\sum_{i=N+1}^\infty M_{i-1}\right)-\sum_{i=N+1}^\infty iM_i\nonumber\\
&=kxNM_N+kxw-v.\label{vdot}\\
\dot{w}&=\sum_{i=N}^\infty \dot{M}_i = kx\sum_{i=N}^\infty (M_{i-1}-M_i)-\sum_{i=N}^\infty M_i\nonumber\\
&=kxM_{N-1}-w.\label{wdot}
\end{align}
Using our new variables in the $x$ equation of \eqref{syst2}, we can write that system augmented with \eqref{udot}, \eqref{vdot} and \eqref{wdot} as
\begin{equation}\label{syst3}
\left\{
\begin{aligned}
\dot{x}&=\alpha-kxu+v\\
\dot{u}&=r-u\\
\dot{v}&=-v+kxw+kNM_N\\
\dot{w}&=-w+kxM_{N-1}\\
\dot{M}_0&=r-M_0-kxM_0\\
\dot{M}_i&=-M_i-kxM_i+kxM_{i-1},\qquad i\geqslant 1.
\end{aligned}
\right.
\end{equation}
We now observe that if we discard the equations for $\dot{M}_i,$ with $i\geqslant N+1,$
we obtain a closed system in the $N+5$ variables $x,u,v,w,M_0,\dots,M_N.$ If we solve this ODE, then, by using the computed $x$ and $M_N,$ all the remaining variables $M_i$ can be recursively computed. Therefore, by defining,
\[
U_1:=x,\quad U_2:=u,\quad U_3:=v,\quad U_4:=w,\quad U_i:=M_{i-5},\quad 5\leqslant i\leqslant N+5,
\]
we can write that finite dimensional system in the form
\begin{equation}\label{NonlinearODE}
\dot{U}=F(U)
\end{equation}
with $F: \Rb^{N+5}\to\Rb^{N+5}$ defined by
\begin{equation}
F(U):=
\left[
\begin{aligned}
\;\alpha&+U_3-kU_1U_2\\
r&-U_2\\
&-U_3+kU_1U_4+NkU_1U_{N+5}\\
&-U_4+kU_1U_{N+4}\\
r&-U_5-kU_1U_5\\
&-U_6+kU_1U_5-kU_1U_6\\
&\qquad\qquad\qquad\vdots\\
&-U_{N+5}+kU_1U_{N+4}-kU_1U_{N+5}\;
\end{aligned}\right].
\end{equation}

\medskip

Let $\Ueq:=(\xeq,\ueq,\veq,\weq,\Meq_0,\dots,\Meq_N)$ be $U$ corresponding to one of the equilibrium solutions mentioned in the previous sections.
Following \cite{cdg} (see previous section), we introduce the variable
\[
y:=\frac{\xeq}{\xeq+1/k},
\]
and for the sake of simplifying notation (and since in the following we will only be referring to the equilibrium quantities) 
we drop the `eq' superscript for the computations in the remaining of this section. 
Hence, for each one of the equilibrium solutions, using the results of the previous section, we have:
\begin{align*}
U_1&:=x=\frac{1}{k}\frac{y}{1-y},\\
U_2&:=u=r,\\
U_3&:=v=r\frac{y^{N+1}}{1-y}\big((N+1)-Ny\big),\\
U_4&:=w=ry^N,\\[.3ex]
U_{5+i}&:=M_i=r(1-y)y^i, \qquad 0\leqslant i\leqslant N.
\end{align*}

 To study the linear stability of  these equilibria $U$ of the ordinary differential equation \eqref{NonlinearODE} we 
have to compute the characteristic polynomial of the 
$(N+5)\times(N+5)$ jacobian matrix $A:=DF(U),$ which is the goal of the next lemma.

Let us introduce the variable
\begin{equation}
\Delta:=1+\lambda(1-y).
\label{eq:delta}
\end{equation}

\begin{lemma}\label{lemmapolycaract}
The characteristic polynomial of the jacobian matrix \( A \) of \eqref{NonlinearODE} 
about an equilibrium  is given by
\begin{multline}\label{polycaract}
\det(A-\lambda I_{N+5})=(-1)^Nkr(1+\lambda)^2(1-y)^{-N-1}\bigg\{(1+\lambda)\left(1+\frac{\lambda}{kr}\right)\Delta^{N+1}\\
-y^N\Big[\Delta^{N+1}+(1-y)\Big(\big(N(1-y)(1+\lambda)+1\big)(1+\dots+\Delta^N)-1\Big)\Big]
\bigg\}.
\end{multline}
\end{lemma}

\begin{proof}
Observing that
\[-1-kU_1=-\frac{1}{1-y},\quad\text{ and }\quad U_{i+4}-U_{i+5}=r(1-y^2)y^{i-1},\]
we can write the \( (N+5)\times(N+5) \) linearization matrix $A =DF(U)$ in the form $A=\left[\begin{array}{@{}c|c@{}}
B&C\\\hline D&E
\end{array}\right]$, where \( B \) is the \( 4\times 4 \) matrix 
\[B=\begin{bmatrix}
-kr&-\dfrac{y}{1-y}&1&0\\[1.5ex]
0&-1&0&0\\
kry^{N}\big(1+N(1-y)\big)&0&-1&\dfrac{y}{1-y}\\[1.5ex]kry^{N-1}(1-y)&0&0&-1
\end{bmatrix},
\]
\( C \) and \(D\) are, respectively, the  \( 4\times(N+1) \) and \( (N+1)\times 4\) matrices
\[C=\begin{bmatrix}
0&\cdots&0&0\\
0&\cdots&0&0\\
0&\cdots&0&\dfrac{Ny}{1-y}\\
0&\cdots&\dfrac{y}{1-y}&0
\end{bmatrix},
\qquad
D=\begin{bmatrix}
-kr(1-y)&0&0&0\\[1ex]
kr(1-y)^{2}&\vdots&\vdots&\vdots\\[1ex]
kr(1-y)^{2}y&\vdots&\vdots&\vdots\\
\vdots&\vdots&\vdots&\vdots\\
kr(1-y)^{2}y^{N-1}&0&0&0
\end{bmatrix},
\]
and 
\( E \) is the  \( (N+1)\times(N+1) \) matrix
\[E=\begin{bmatrix}[1.8]
-\frac{1}{1-y}&0&\cdots&\cdots&0\\[1ex]
\frac{y}{1-y}&-\frac{1}{1-y}& & &\vdots\\
0&\ddots&\ddots & &\vdots\\
\vdots& &\ddots & \ddots& 0\\
0&\dots& 0 &\frac{y}{1-y}&-\frac{1}{1-y}
\end{bmatrix}.
\]
To compute the determinant of \( A-\lambda I_{N+5} \)
we will take advantage of the particular structure of the matrix \( A \) pointed out above and 
start by writing
\[ \det(A-\lambda I_{N+5})=\det\left[\begin{array}{@{}c|c@{}}
B-\lambda I_{4}&C\\\hline D&E-\lambda I_{N+1}
\end{array}\right].\]

\medskip

Then, we successively perform the following operations in \( A-\lambda I_{N+5} \) to achieve a final matrix with equal determinant:
\begin{enumerate}
\item factor out $\frac{1}{1-y}$ from the last $N+1$ columns;
\item factor out $\frac{1}{1-y}$ from the third row;
\item factor out  $kr(1-y)$ from the first column;
\item apply Laplace determinant expansion relative to the second row.
\end{enumerate}
In the end we obtain,
\begin{equation}\label{detalambda0}
 \det(A-\lambda I_{N+5})=-kr(1+\lambda)(1-y)^{-N-1}\det(\widetilde{A}_{\lambda}),
\end{equation}
with $\widetilde{A}_\lambda=\left[\begin{array}{@{}c|c@{}}
\widetilde{B}_\lambda&\widetilde{C}\\\hline\\[-0.9em] \widetilde{D}&\widetilde{E}_\lambda
\end{array}\right],$ where, $\widetilde{B}_\lambda$ is the $3\times 3$ matrix
\[\widetilde{B}_\lambda=\begin{bmatrix}
-(1+\lambda/kr)(1-y)^{-1}&1&0\\[1.5ex]
y^{N}\big(1+N(1-y)\big)&-(1+\lambda)(1-y)&y\\[1.5ex]y^{N-1}&0&-(1+\lambda)
\end{bmatrix},
\]
\( \widetilde{C} \) and \(\widetilde{D}\) are, respectively, the  \( 3\times(N+1) \) and \( (N+1)\times 3\) matrices
\[\widetilde{C}=\begin{bmatrix}
0&\cdots&0&0\\
0&\cdots&0&Ny(1-y)\\
0&\cdots&y&0
\end{bmatrix},
\qquad
\widetilde{D}=\begin{bmatrix}
-1&0&0\\[1ex]
1-y&\vdots&\vdots\\[1ex]
(1-y)y&\vdots&\vdots\\
\vdots&\vdots&\vdots\\
(1-y)y^{N-1}&0&0
\end{bmatrix},
\]
and $\widetilde{E}_\lambda$ is the $(N+1)\times(N+1)$ matrix,
\[\widetilde{E}_\lambda=\begin{bmatrix}[1.5]
-\Delta&0&\cdots&\cdots&0\\
y&-\Delta& & &\vdots\\
0&\ddots&\ddots & &\vdots\\
\vdots& &\ddots & \ddots& 0\\
0&\dots& 0 &y&-\Delta
\end{bmatrix}.
\]
The main idea here is to left multiply $\tilde{A}_\lambda$ by a square matrix of determinant 1, 
in such a way that the resulting matrix has a more easily computable determinant. 
For the following we consider that $\Delta\not= 0$. Consider the $(N+1)\times(N+1)$ matrix
 \[\Lambda=\begin{bmatrix}[1.0]
1&0&\cdots&\cdots&0\\
\frac{y}{\Delta}&1& & &\vdots\\
\frac{y^2}{\Delta^2}&\frac{y}{\Delta}&1 & &\vdots\\
\vdots&\ddots &\ddots & \ddots& 0\\
\frac{y^N}{\Delta^N}&\frac{y^{N-1}}{\Delta^{N-1}}& \cdots &\frac{y}{\Delta}&1
\end{bmatrix},
\]
and observe that $\Lambda\widetilde{E}_\lambda=-\Delta I_{N+1}.$ Therefore,
\begin{equation}\label{detAlambda}
\det(\widetilde{A}_{\lambda})=
\det \left( \left[
\begin{array}{@{}c|c@{}}
I_3 & 0\\
\hline\\[-0.9em]
0 & \Lambda
\end{array}
\right]\widetilde{A}_{\lambda}
\right)=
\det \left[
\begin{array}{@{}c|c@{}}
\widetilde{B}_\lambda & \widetilde{C}\\
\hline \\[-0.9em]
\Lambda \widetilde{D} & -\Delta I_{N+1}
\end{array}
\right].
\end{equation}
The second and third columns of $\Lambda \widetilde{D}$ are $N+1$ dimension nul columns, while 
the first column of $\Lambda \widetilde{D}=
\big[-1 \quad \alpha_0\quad \alpha_1\quad \cdots\quad \alpha_{N-1}\big]^{\top},
$
where, for $0\leqslant i\leqslant N-1,$ these entries are given by,
\[
\alpha_i:=
\left[\frac{y^{i+1}}{\Delta^{i+1}}\quad \frac{y^{i}}{\Delta^{i}}\quad \cdots\quad 1 \quad 0 \quad\cdots\quad  0\right]
\begin{bmatrix}
-1\\[1ex]
1-y\\[1ex]
(1-y)y\\
\vdots\\
(1-y)y^{N-1}
\end{bmatrix}.
\]
In particular we will need explicit expressions for the last two entries:
\begin{equation}\label{alphaN}
\begin{aligned}
\alpha_{N-2}&=-\frac{y^{N-1}}{\Delta^{N-1}}+(1-y)\frac{y^{N-2}}{\Delta^{N-2}}(1+\Delta+\dots+\Delta^{N-2}),\\
\alpha_{N-1}&=-\frac{y^N}{\Delta^N}+(1-y)\frac{y^{N-1}}{\Delta^{N-1}}(1+\Delta+\dots+\Delta^{N-1}).
\end{aligned}
\end{equation}
Now, we eliminate the nonzero entries of $\widetilde{C}$. By multiplying the last row of the last matrix in \eqref{detAlambda} by $\frac{Ny(1-y)}{\Delta}$ and adding to the second row, we eliminate the last entry of this row. Then, by multiplying the penultimate row by $\frac{y}{\Delta}$ and adding to the third row, we eliminate the penultimate  entry of this row. Therefore, from \eqref{detAlambda}
\begin{equation}\label{detalambda1}
\det(\widetilde{A}_{\lambda})=
\det \left[
\begin{array}{@{}c|c@{}}
\widetilde{B}^*_\lambda & 0\\
\hline \\[-0.9em]
\Lambda \widetilde{D} & -\Delta I_{N+1}
\end{array}\right]
=(-1)^{N+1}\Delta^{N+1}\det \widetilde{B}^*_\lambda, 
\end{equation}
where,
\[\widetilde{B}^*_\lambda=\begin{bmatrix}
-(1+\lambda/kr)(1-y)^{-1}&1&0\\[1.5ex]
y^{N}\big(1+N(1-y)\big)+\alpha_{N-1}\frac{Ny(1-y)}{\Delta}&-(1+\lambda)(1-y)&y\\[1.5ex]y^{N-1}+\alpha_{N-2}\frac{y}{\Delta}&0&-(1+\lambda)
\end{bmatrix},
\]
and therefore,
\begin{equation}\label{detBlambda}
\begin{aligned}
\det \widetilde{B}^*_\lambda=&-(1+\lambda)^2\left(1+\frac{\lambda}{kr}\right)\\
&+\left(y^N\big(1+N(1-y)\big)+\alpha_{N-1}\frac{Ny(1-y)}{\Delta}\right)(1+\lambda)\\
&+y^N+\alpha_{N-2}\frac{y^2}{\Delta}.
\end{aligned}
\end{equation}
Using the explicit expressions for $\alpha_{N-1}$ and $\alpha_{N-2}$ given by \eqref{alphaN} we get
\begin{align*}
y^N\big(1+N(1-y)\big)&+\alpha_{N-1}\frac{Ny(1-y)}{\Delta}=
y^N\left[1+N(1-y)\left(1+\frac{y^{1-N}}{\Delta}\alpha_{N-1}\right)\right]\\
&=y^N\left[1+\frac{N(1-y)}{\Delta^{N+1}}(\Delta-y)(1+\dots+\Delta^N)\right],
\end{align*}
and 
\begin{align*}
y^N+\alpha_{N-2}\frac{y^2}{\Delta}
&=y^N\left[1+\frac{1}{\Delta^N}\left(-y+(1-y)(\Delta+\dots+\Delta^{N-1})\right)\right]\\
&=\frac{y^N}{\Delta^{N}}(\Delta-y)(1+\dots+\Delta^{N-1})\\
&=\frac{y^N}{\Delta^{N+1}}\left[(\Delta-y)(1+\dots+\Delta^{N})-(\Delta-y)\right].
\end{align*}
By plugging this last expression in \eqref{detBlambda} we have,
\begin{align*}
\det \widetilde{B}^*_\lambda=&-(1+\lambda)^2\left(1+\frac{\lambda}{kr}\right)\\
&+y^N\Big\{1+\lambda+\frac{N(1-y)(1+\lambda)}{\Delta^{N+1}}(\Delta-y)(1+\dots+\Delta^N)\\
&+\frac{1}{\Delta^{N+1}}\left[(\Delta-y)(1+\dots+\Delta^{N})-(\Delta-y)\right]\Big\}\\
=&-(1+\lambda)^2\left(1+\frac{\lambda}{kr}\right)\\
&+y^N(1+\lambda)\Big\{1+\frac{(1-y)}{\Delta^{N+1}}\left[\left(N(1-y)(1+\lambda)+1\right)(1+\dots+\Delta^N)-1\right]\Big\},
\end{align*}
where we have used the fact that $\Delta-y=(1+\lambda)(1-y).$ By using \eqref{detalambda0}, \eqref{detalambda1} and last equation, we obtain \eqref{polycaract}.
\end{proof}

\medskip

The determinant of the matrix $A$ is obtained from \eqref{polycaract} by making $\lambda=0,$ in which case, also $\Delta=1$ and we obtain:
\begin{equation}\label{detA}
\det A=(-1)^Nkr(1-y)^{-N-1}P_{N}(y).\end{equation}
Therefore, if we compare this with \eqref{pNy} we see that the bifurcation condition $p_N(y)=0$ is equivalent to
$\det A=0,$ as it should be. Let $y^*=y^*(N)$ be the unique solution of this bifurcation equation to which corresponds $\alpha/r=\mu^*=\mathcal{F}_N(y^*).$ 
Our next step is to show that $\lambda=0$ is a simple eigenvalue of $A,$ when $y=y^*$. 

\begin{lemma}\label{lemmasimple}
For $y=y^*$, $\lambda=0$ is a simple eigenvalue of $A.$
\end{lemma}

\begin{proof}
First, for a generic equilibrium, and therefore for a generic $y\in(0,1),$ we compute $a_1$ (depending on $N,kr,y$) such that, as $\lambda\to 0,$
\begin{equation}\label{cpdev}
\det(A-\lambda I_{N+5})=\det A +a_1\lambda +O(|\lambda|^2),
\end{equation}
with fixed $N, kr, y.$

It is convenient to introduce
\[g(\Delta):=\sum_{i=0}^N \Delta^i.\]
Therefore, since $g(1)=N+1,$ and $g'(1)=\frac{N(N+1)}{2},$ we obtain, as $\lambda\to 0,$
\[g(\Delta)=g_0+g_1\lambda+O(|\lambda|^2),\]
for $g_0=N+1$ and $g_1=\frac{1}{2}N(N+1)(1-y).$ Hence, in \eqref{polycaract}, we will have
\begin{align*}
\big(N(1-&y)(1+\lambda)+1\big)g(\Delta)\\
&=\big[N(1-y)+1+N(1-y)\lambda\big]\big(g_0+g_1\lambda\big)+O\left(|\lambda|^2\right)\\
&=\big[N(1-y)+1\big] g_0+\big[\big(N(1-y)+1\big)g_1+N(1-y)g_0\big]\lambda+O\left(|\lambda|^2\right)\\
&=b_0+b_1\lambda+O\left(|\lambda|^2\right),
\end{align*}
where,
\begin{align*}
b_0&:=(N+1)\big(N(1-y)+1\big)\\
b_1&:=\frac{1}{2}N(N+1)(1-y)\big(N(1-y)+3\big).
\end{align*}
Taking in account that,
\[\Delta^{N+1}=1+(N+1)(1-y)\lambda+O\left(|\lambda|^2\right),\]
so that
\[
\left[(1+\lambda)\left(1+\frac{\lambda}{kr}\right)-y^N\right]\Delta^{N+1}=c_0+c_1\lambda+O\left(|\lambda|^2\right),
\]
where,
\begin{align*}
c_0&=1-y^N\\
c_N&=1+\frac{1}{kr}+(N+1)(1-y)(1-y^N),
\end{align*}
we have in \eqref{polycaract},
\begin{align*}
(1+\lambda)\left(1+\frac{\lambda}{kr}\right)\Delta^{N+1}&\\
-y^N\Big[\Delta^{N+1}&+(1-y)\Big(\big(N(1-y)(1+\lambda)+1\big)(g(\Delta)-1\Big)\Big]\\
&=d_0+d_1\lambda+O\left(|\lambda|^2\right),
\end{align*}
where,
\begin{align}
d_0&=1-y^N\Big[1+(1-y)\Big(\big(N(1-y)+1\big)(N+1)-1\Big)\Big]\label{d0}\\
d_1&=1+\frac{1}{kr}+(N+1)(1-y)\left\{1-y^N\left[1+\frac{N(1-y)}{2}\big(N(1-y)+3\big)\right]\right\}.\label{d1}
\end{align}
Therefore, by \eqref{polycaract}
\begin{align*}
\det(A-\lambda I_{N+5})&=(-1)^Nkr(1-y)^{-N-1}\big(d_0+(2d_0+d_1)\lambda\big)+O\left(|\lambda|^2\right).
\end{align*}
Now, by comparing \eqref{d0} and \eqref{pNy}, we observe that $d_0=p_N(y)$, so that,
\[\det(A-\lambda I_{N+5})=\det A+\left[2\det A+(-1)^Nkr(1-y)^{-N-1}d_1\right]\lambda+O\left(|\lambda|^2\right).\]
Therefore, we obtain \eqref{cpdev} with
\[a_1=2\det A+(-1)^Nkr(1-y)^{-N-1}d_1.\]
Now, when we are considering the equilibrium corresponding to $(\alpha/r,y)=(\mu^*,y^*)$ we know that
$\det A=0$, so that, as $\lambda\to 0,$
\begin{align*}
\det(A-\lambda I_{N+5})&=(-1)^Nkr(1-y^*)^{-N-1}d_1\lambda+O\left(|\lambda|^2\right).
\end{align*}
Hence, $\lambda=0$ will be a simple eigenvalue of $A$ if and only if $d_1\not=0$ for $y=y^*$,
what we are going to show that indeed it is here the case. Let us define in \eqref{d1},
\begin{align*}
q_N(y)&:=1-y^N\left[1+\frac{N(1-y)}{2}\big(N(1-y)+3\big)\right]\\
&=1-y^N\left[1+\frac{3}{2}N(1-y)+\frac{1}{2}N^2(1-y)^2\right].
\end{align*}
Rewriting $p_N(y)$ in the form
\[
p_N(y)=1-y^N\Big[1+N(1-y)+N(N+1)(1-y)^2\Big],
\]
we easily obtain
\[q_N(y)-p_N(y)=\frac{1}{2}y^N(1-y)N\Big[(N+2)(1-y)-1\Big].\]
Now, consider the case $y=y^*.$ Since by definition, $p_N(y^*)=0$, we get, 
\[q_N(y^*)=\frac{1}{2}(y^*)^N(1-y)N\Big[(N+2)(1-y^*)-1\Big].\] 
But according to \cite{cdg} (see previous section), we know that $0<y^*<\frac{N+1}{N+2},$ so that,
\[
(N+2)(1-y^*)-1>(N+2)\left(1-\frac{N+1}{N+2}\right)-1=0,
\]
which proves that, for $y=y^*$, $q_N(y^*)>0$, and therefore, $d_1>0$. 
This completes the proof that, for $y=y^*$, $\lambda=0$ is a simple eigenvalue of $A$.
\end{proof}

\medskip

The next lemma will be crucial for the stability result in Theorem~\ref{stabilityatbifpoint}

\medskip

\begin{lemma}\label{imaginary}
For every $0<y\leqslant y^*$, the matrix $A$ does not have pure imaginary eigenvalues.
\end{lemma}
\begin{proof}
We intend to prove that, if $0< y\leqslant y^*,$ then, the equation $\det(A-\lambda I_{N+5})=0$ does not have pure imaginary solutions. Using \eqref{polycaract}, this equation, for $kr\not=0$ and $\lambda\not=-1$, is equivalent to
$$
(1+\beta\lambda)y^{-N}
=\Big[1+(1-y)\Big(\big(N(1-y)(1+\lambda)+1\big)g(\Delta)-1\Big)\Delta^{-N-1}\Big](1+\lambda)^{-1},
$$
recalling that,
$
g(\Delta):=\sum_{i=0}^N\Delta^i,
$
and defining $\beta:=\frac{1}{kr}$. By writing,
\begin{align*}
F_N(\lambda,\beta,y)&:=(1+\beta\lambda)y^{-N},\\
G_N(\lambda,y)&:=\Big[1+(1-y)\Big(\big(N(1-y)(1+\lambda)+1\big)g(\Delta)-1\Big)\Delta^{-N-1}\Big](1+\lambda)^{-1},
\end{align*}
the above equation can be written as
\begin{equation}\label{FNGN}
F_N(\lambda,\beta,y)=G_N(\lambda,y).
\end{equation}
Now, take $\lambda=i\omega,$ with real $\omega\not=0.$ Then, since $\beta>0,$
\[
|F_N(i\omega,\beta,y)|=y^{-N}\sqrt{1+\beta^2\omega^2}> y^{-N}.
\]
On the other hand, defining
$$
\hat{g}(\Delta):=\sum_{i=1}^N\Delta^{-i},
$$
we can write,
\begin{align*}
G_N(\lambda,y)=\big(1+(1-y)\Delta^{-1}\hat{g}(\Delta)\big)(1+\lambda)^{-1}+N(1-y)^2(\hat{g}(\Delta)+1)\Delta^{-1}.
\end{align*}
By observing that, for $\lambda=i\omega,$ with real $\omega\not=0,$ we have $|(1+\lambda)^{-1}|<1$, but also $|\Delta^{-1}|<1,$ which in turn implies,
$
|\hat{g}(\Delta)|<N,
$
we conclude that,
$$
|G_N(i\omega,y)|<1+N(1-y)+N(N+1)(1-y)^2,
$$
so that,
$$
|F_N(i\omega,\beta,y)|-|G_N(i\omega,y)|>y^{-N}-\Big[1+N(1-y)+N(N+1)(1-y)^2\Big]=y^{-N}p(y).
$$
But recalling the results summarized in section 2., we know that, for $0<y\leqslant y^*,$ $p_N(y)\geqslant 0,$ and therefore,
$$
|F_N(i\omega,\beta,y)|-|G_N(i\omega,y)|>0
$$
which makes it impossible for equation \eqref{FNGN} to be satisfied for any $\lambda=i\omega$, with $\omega\not=0.$
\end{proof}

\medskip

We can now state the main result of this section:

\medskip

\begin{theorem}\label{stabilityatbifpoint}
Let $\mu^*$ be as in Proposition~\ref{prop:structure}, and let $x^*$ be the value of the $U_1$ 
component of the unique equilibrium of \eqref{NonlinearODE}
when $\alpha/r = \mu^*$ (see Fig.~\ref{figbif}.) Then, for every $\alpha/r \in (0, \mu^*)$ and all $kr>0$,  
the equilibrium solution $U^{1*}$, with $U_1^{1*}<x^{*}$,  is locally exponentially asymptotically stable, 
and the equilibrium solution $U^{2*}$, with $U_1^{2*}>x^{*}$, is unstable.
\end{theorem}

\begin{proof}
Let us rewrite \eqref{polycaract} as follows:
\begin{align}\label{polycaractalt}
\det(A-\lambda I_{N+5})  =
& (-1)^N(1-y)^{-N-1}(1+\lambda)^3\lambda\Delta^{N+1}  - \\
& - kr(-1)^N(1-y)^{-N-1}(1+\lambda)^2\times \nonumber   \\
&\quad \times \biggl\{ (1+\lambda)\lambda\Delta^{N+1} - y^N\Big[\Delta^{N+1}+\nonumber   \\
& \quad\qquad+(1-y)\Big(\big(N(1-y)(1+\lambda)+1\big)(1+\dots+\Delta^N)-1\Big)\Big]\biggr\}.  \nonumber
\end{align}
Observe that, if $kr=0$, then  $\det(A-\lambda I_{N+5})  =  (-1)^N(1-y)^{-N-1}(1+\lambda)^3\lambda\Delta^{N+1}$, and
the eigenvalues of $A$ are $\lambda = 0$ (simple), $\lambda = -1$ (with algebraic multiplicity 3,) and $\lambda = -\frac{1}{1-y}$
(with algebraic multiplicity $N+1$.) Let $y=y^*$. Then, from Lemmas~\ref{lemmasimple} and \ref{imaginary}, for
every $kr>0$ the linearization of \eqref{NonlinearODE} around the equilibrium $U^*$ with $y=y^*$ has $N+4$ 
nonzero eigenvalues with negative real parts and the remaining eigenvalue $\lambda =0$ is simple.

\medskip

For $0<\alpha/r<\mu^*$  let $y^{1*}<y^*<y^{2*}$ be the only two values of 
$y$ that solve the bifurcation equation \eqref{bifeqy}. 
To these values of $y$  corresponds two equilibria of \eqref{NonlinearODE}: $U^{1*}$ 
(corresponding to $y^{1*}$) and $U^{2*}$ (corresponding to $y^{2*}.$)
By what was done previously, in particular from \eqref{Fprime}, \eqref{pNy}, \eqref{detA},  \eqref{cpdev}, \eqref{d0}, 
and \eqref{d1}, the jacobian matrix of the linearization of \eqref{NonlinearODE} around 
$U^{j*}$ has eigenvalues given by the solutions $\lambda$ of
\begin{equation}\label{spectraUj}
p_{N}(y^{j*})+ 
(2 p_{N}(y^{j*}) + d_1)\lambda + O(|\lambda|^2) = 0 
\quad\text{as $\lambda \to 0$},
\end{equation}
with $d_1=d_1(y^{j*})$ given by \eqref{d1}.
From the study of equilibria in \cite{cdg}, recalled in Section~\ref{sec2}, we know that
$p_N(y^{1*})>0$  for all $y^{1*}<y^*$,
 and  $p_N(y^{2*})<0$ for all $y^{2*}>y^*$.
From the proof above we have $d_1(y^*)>0$ and hence, by continuity,  for $y^{1*}$ and $y^{2*}$ sufficiently 
close to $y^*$  it still holds that 
$2 p_N(y^{j*}) + d_1(y^{j*}) >0.$
This implies that, for sufficiently small $kr>0$, equation \eqref{spectraUj} has a negative solution  
when $j=1$ and a positive solution when $j=2$. 

\medskip

Thus, from the argument above, the zero eigenvalue of the jacobian matrix at
the bifurcation point $y^*$ is perturbed to a negative eigenvalue for the
linearization about the equilibrium $U^{1*}$ when $y^{1*}$ is close to $y^*$.
By Lemma~\ref{imaginary} all the other eigenvalues of the Jacobian at $U^{1*}$ have negative real parts, and
since $\lambda=0$ is not an eigenvalue if $y$ is not equal to $y^*$, we conclude that for all equilibria $U^{1*}$ 
(not necessarily close to $U^*$) the real negative eigenvalue originated from $\lambda=0$ at the bifurcation point
cannot become nonnegative. Hence,
for all values of the parameters $\alpha/r\in (0, \mu^*)$, $kr>0$, the equilibrium  $U^{1*}$ of
\eqref{NonlinearODE} is locally exponentially asymptotically stable.

\medskip

As in the case  of $U^{1*}$ above, when $U^{2*}$ is a sufficiently small perturbation of $U^*$, the zero
eigenvalue of the corresponding jacobian is perturbed to a positive real eigenvalue, and, 
by continuity, all other eigenvalues have negative real parts if the perturbation is sufficiently small.
Also, this positive eigenvalue cannot become nonpositive if $y$ remains larger than $y^*$. This implies that,
for all values of the parameters $\alpha/r\in (0, \mu^*)$, $kr>0$, the equilibrium  $U^{2*}$ of
\eqref{NonlinearODE} is unstable. 

\medskip

This completes the proof of the theorem.
\end{proof}

\medskip

\begin{remark}\label{remark}
In the instability part of the previous proof  we establish that the eigenvalue 
of the jacobian matrix at $U^{2*}$ that becomes positive when $U^{2*}$ is a small perturbation of $U^*$
cannot become nonpositive for larger perturbations (i.e., for larger positive values of $y-y^*$).
However, note that for these equilibria with $y>y^*$ we could not prove
a result analogous to Lemma~\ref{imaginary} and so we cannot guarantee that, by changing the system's parameters,
one or more pairs of complex conjugated eigenvalues will not cross the imaginary axis
from left to right thus increasing the dimension of the unstable manifold. 
Numerical evidence, some presented in section~\ref{sec5}, lead us to conjecture that this is not the case. 
\end{remark}

\begin{conj}\label{conjonstability}
With the assumptions and notation of Theorem~\ref{stabilityatbifpoint} we have that
for all $\alpha/r\in (0,\mu^*)$ and all $kr>0$,
 the unstable manifold of all equilibria  
\( U^{2*} \) has dimension one.
\end{conj}

\section{Local dynamics of the silicosis system \eqref{syst2}}\label{sec4}

\begin{theorem}\label{strongstability}
Let $\alpha/r<\mu^*$ and let $\widetilde{U}^{\text{eq}} =(x^{\text{eq}} , M_0^{\text{eq}}, M_1^{\text{eq}},\ldots) $ be an equilibrium solution of
\eqref{syst2} such that the corresponding equilibrium  of the $(N+5)$-dimensional system \eqref{NonlinearODE},
$U^{\text{eq}}=(U_1^{\text{eq}}, \ldots, U_{N+5}^{\text{eq}})$,
is locally exponentially
asymptotically stable. Then, $\widetilde{U}^{\text{eq}}$ is a locally asymptotically stable solution of  \eqref{syst2} in
the strong topology of $X.$
\end{theorem}

\begin{proof}
Remember that the silicosis system \eqref{syst2} is 
equivalent to the infinite system \eqref{syst3} with restrictions \eqref{uvw}. 
To every point $\widetilde{U}=(x,M_0, M_1, \ldots)\in X_+$ there corresponds a unique 
$U = (x,u,v,w,M_0, \ldots, M_N)\in\Rb^{N+5}_+$.
By what was done in section~\ref{sec2} we
know that there exists an open set $\Omega\subset \Rb^{N+5}$ containing $U^{\text{eq}}$ such that for
every initial condition in $\Omega$ the corresponding solution of the 
$(N+5)$-dimensional system \eqref{NonlinearODE} converges to $U^{\text{eq}}$ when $t\to+\infty.$ 
In particular, for those initial conditions, we have that $x(t)\to x^{\text{eq}}$ and  $M_i(t)\to M_i^{\text{eq}}$ as 
$t\to+\infty$ for all $i=0, \ldots, N.$
Using this  in the equations  in \eqref{syst2} for $M_i$ with $i > N$ we conclude that all components 
of the solution $\widetilde{U}=(x, M_0, M_1, \ldots)$ 
converge exponentially
to the corresponding components of $\widetilde{U}^{\text{eq}}=(x^{\text{eq}}, M_0^{\text{eq}}, M_1^{\text{eq}},\ldots)$ 
when $t\to +\infty$.

\medskip

Observe that, from the definition of the variables $u$ and $v$ in \eqref{uvw}, if
 $\widetilde{U}=(x, M_0, M_1, \ldots)$ is a nonnegative
solution of \eqref{syst2} in $[0,+\infty)$, then, for all $t\geqs 0$, the norm of $\widetilde{U}(t)$ 
can be written in the form
\begin{equation}
\| \widetilde{U}(t)\| = x(t) + u(t) + v(t) + \sum_{i=0}^NiM_i(t).\label{normU}
\end{equation}
Let $B_\varepsilon\subset X_+$ be an open ball of radius $\varepsilon$ centered at the equilibrium
$\widetilde{U}^{\text{eq}}.$ Take an initial condition $\widetilde{U}(0)\in B_\varepsilon$. Then, since
\[
\Bigl|\| \widetilde{U}(0)\| - \|\widetilde{U}^{\text{eq}}\|\Bigr| \leqs \|\widetilde{U}(0)-\widetilde{U}^{\text{eq}}\| < \varepsilon,
\]
the equality \eqref{normU} with $t=0$ implies that, if we choose $\varepsilon$ small enough, the corresponding
initial condition $U(0)$ for the $(N+5)$-dimensional system \eqref{NonlinearODE} will be in $\Omega.$

\medskip

Hence, for small enough $\varepsilon$, to every initial condition $\widetilde{U}(0) \in 
B_\varepsilon\subset X_+$  corresponds a vector
$U(0) \in \Rb^{N+5}_+$  in $\Omega$, and so, the solution $\widetilde{U}(\cdot)$ of
\eqref{syst2} satisfies
$\|\widetilde{U}(t)\| \to \|\widetilde{U}^{\text{eq}}\|$ exponentially as $t\to +\infty$. 
This, together
with the componentwise convergence of $\widetilde{U}$ to $\widetilde{U}^{\text{eq}}$, implies, by a standard result 
(see, e.g., \cite[Lemma 3.3]{bcp}), 
 that $\widetilde{U}(t)\to \widetilde{U}^{\text{eq}}$ strongly in $X$ as $t\to +\infty.$\end{proof}

\medskip

\begin{theorem}\label{expstrongstability}
Under the assumptions of Theorem~\ref{strongstability} the solutions of  \eqref{syst2} that converge
in the strong topology of $X$  to the 
locally asymptotically stable solution $\widetilde{U}^{\text{eq}}$   do so at an exponential rate.
\end{theorem}

\begin{proof}
Let $\widetilde{U}(t)\to \widetilde{U}^{\text{eq}}$ in $X$, as $t\to +\infty$.
We know that each component of $\widetilde{U}(t)$ converges exponentially to the corresponding
component of $\widetilde{U}^{\text{eq}}$. To prove the theorem we need to show 
that $\|\widetilde{U}(t)-\widetilde{U}^{\text{eq}}\|$ converges exponentially fast to zero as $t\to +\infty$.
First, we have to prove the same holds for the $\ell^1$ norm.

\medskip

From \eqref{syst2}, we obtain 
\[
\frac{d}{dt}(M_0-M_0^{\text{eq}}) =  - (1+kx^{\text{eq}})(M_i-M_i^{\text{eq}}) - kM_0(x-x^{\text{eq}}),
\]
and, for each \(i\geq 1\),
\begin{align*}
\frac{d}{dt}(M_i-M_i^{\text{eq}})&= - (1+kx^{\text{eq}})(M_i-M_i^{\text{eq}})\\
&+kx^{\text{eq}}(M_{i-1}-M_{i-1}^{\text{eq}}) - k(M_i-M_{i-1})(x-x^{\text{eq}}).
\end{align*}
For each $t>0$ and integer $i\geqslant 0$ define
\[
\delta_i(t):=M_i(t)-M_i^{\text{eq}}(t),\qquad 
\varphi_i(t):=\bigl(x^{\text{eq}}\bigr)^{-1}M_i(t)\bigl(x(t)- x^{\text{eq}}\bigr).
\]
Changing the time variable $t\mapsto\eqref{NonlinearODE}$, denoting by $(\cdot)'$ the derivarive $\frac{d}{d\tau}$, and
defining $\beta := \frac{1}{kx^{\text{eq}}}$,
the system above can be written as
\begin{equation}\label{dotdelta0}
\delta'_0=-(1+\beta)\delta_0 - \varphi_0(\tau)\,
\end{equation}
and
\begin{equation}\label{dotdeltai}
\delta'_i=-(1+\beta)\delta_i+\delta_{i-1} - \varphi_i(\tau) + \varphi_{i-1}(\tau),\qquad i=1,2,\dots
\end{equation}

\medskip

Note that system \eqref{dotdelta0}-\eqref{dotdeltai} can be solved recursively, starting with the equation for $\delta_0$
and then sequentially for $\delta_i$ for $i=1, 2, \ldots$, since the equation for $\delta_i$ only depends on the
components of the solutions with $j\leqslant i$. So, consider the $(n+1)$-dimensional system for the vector
of displacements {\boldmath${\delta}$}$_n=(\delta_0, \delta_1, \ldots, \delta_n)^{\text{\sf T}},$
\begin{equation}\label{nsystemdelta}
\text{\boldmath${\delta}$}_n' = J_{n+1}
\text{\boldmath${\delta}$}_n + \text{\boldmath${\Phi}$}_n(\tau),
\end{equation}
where $\text{\boldmath${\Phi}$}_n  = (\Phi_i)_{i=0,\ldots,n}^{\text{\sf T}}$ with $\Phi_0= -\varphi_0$ and $\Phi_i=-\varphi_i+\varphi_{i-1}$ if $i\geqs 1$, and
$J_{n+1}$ is the $(n+1)$-dimensional Jordan matrix
\begin{equation}\label{Jmatrix}
J_{n+1} := \begin{bmatrix} 
- (1+\beta) &  &  &  &  \\
1 & - (1+\beta) & & &  \\
   & 1 & - (1+\beta) & & \\
   & & \ddots & \ddots & \\
   & & & 1 & - (1+\beta)
\end{bmatrix}. 
\end{equation}

The solution of \eqref{nsystemdelta} is given by the variation of constants formula
\begin{equation}
\text{\boldmath${\delta}$}_n(\tau) = e^{J_{n+1}\tau}\text{\boldmath${\delta}$}_n(0) + 
\int_0^\tau e^{J_{n+1}(\tau - s)}\text{\boldmath${\Phi}$}_n(s)ds \label{varconst}
\end{equation}
and we now estimate each of the terms in the right-hand side of this expression separately.

\medskip

For the first term in the right-hand side of \eqref{varconst} we have
\begin{align}
e^{J_{n+1}\tau}\text{\boldmath${\delta}$}_n(0) 
& = e^{-(1+\beta)\tau}
\begin{bmatrix}
1 & 0 & 0 &  \cdots & 0 \\
\tau & 1 & 0 & \cdots & 0 \\
\frac{\tau^2}{2!} & \tau & 1 & \cdots & 0 \\
\vdots & \vdots & \vdots &  \ddots & \vdots \\
\frac{\tau^n}{n!} & \frac{\tau^{n-1}}{(n-1)!} & \frac{\tau^{n-2}}{(n-2)!} & \ldots & 1
\end{bmatrix}
\begin{bmatrix}
\delta_0(0) \\ \delta_1(0) \\ \delta_2(0) \\ \vdots \\ \delta_n(0)
\end{bmatrix} \nonumber \\
&  = e^{-(1+\beta)\tau}
\begin{bmatrix}
\delta_0(0)  \\
\tau\delta_0(0)  + \delta_1(0)  \\
\frac{\tau^2}{2!} \delta_0(0)  + \tau \delta_1(0) + \delta_2(0)  \\
\vdots \\ 
\frac{\tau^n}{n!} \delta_0(0)  + \frac{\tau^n}{n!} \delta_0(0) + \ldots + \tau \delta_{n-1}(0) + \delta_n(0) \nonumber
\end{bmatrix}
\end{align}
and hence
\begin{align}
\bigl\|e^{J_{n+1}\tau}\text{\boldmath${\delta}$}_n(0)\bigr\|_{\ell^1} 
& \leqs e^{-(1+\beta)\tau}\sum_{j=0}^n\Bigl(\sum_{k=0}^j\frac{\tau^k}{k!}\left|\delta_{j-k}(0)\right|\Bigr) \nonumber \\
& = e^{-(1+\beta)\tau}\sum_{k=0}^n\frac{\tau^k}{k!}\sum_{p=0}^{n-k}\left|\delta_{p}(0)\right| \nonumber \\
& \leqs e^{-(1+\beta)\tau}e^{\tau}\left\|\text{\boldmath${\delta}$}(0)\right\|_{\ell^1}  
=  e^{-\beta\tau}\left\|\text{\boldmath${\delta}$}(0)\right\|_{\ell^1},\label{bound1}
\end{align}
where $\text{\boldmath${\delta}$}(0) := (\delta_0(0), \delta_1(0), \delta_2(0), \ldots)^{\text{\sf T}}.$

\medskip

For the second term in the right-hand side of \eqref{varconst} we can write
\begin{eqnarray}
\lefteqn{\int_0^\tau e^{J_{n+1}(\tau - s)}\text{\boldmath${\Phi}$}_n(s)ds =}\\
& = & \int_0^\tau  e^{-(1+\beta)(\tau-s)}
\begin{bmatrix}
1 & 0 & 0 & \cdots & 0 \\
(\tau-s) & 1 & 0 & \cdots & 0 \\
\frac{(\tau-s)^2}{2!} & (\tau-s) & 1 & \cdots & 0 \\
\vdots & \vdots & \vdots &  \ddots & \vdots \\
\frac{(\tau-s)^n}{n!} & \frac{(\tau-s)^{n-1}}{(n-1)!} & \frac{(\tau-s)^{n-2}}{(n-2)!} & \ldots & 1
\end{bmatrix}
\begin{bmatrix}
\Phi_0(s) \\ \Phi_1(s) \\ \Phi_2(s)  \\ \vdots \\ \Phi_n(s))
\end{bmatrix} ds \nonumber \\
&  = & \int_0^\tau  e^{-(1+\beta)(\tau-s)}
\begin{bmatrix}
\Phi_0(s)   \\
(\tau - s)\Phi_0(s)   + \Phi_1(s)   \\
\frac{(\tau - s)^2}{2!}\Phi_0(s) + (\tau - s)\Phi_1(s)   + \Phi_2(s)  \\
\vdots \\ 
\frac{(\tau - s)u^n}{n!} \Phi_0(s)  + \frac{(\tau - s)^n}{n!} \Phi_1(s) + \ldots + (\tau - s) \Phi_{n-1}(s) + \Phi_n(s) \nonumber
\end{bmatrix}
ds,
\end{eqnarray}
and hence
\begin{align}
\left\|\int_0^\tau e^{J_{n+1}(\tau - s)}\text{\boldmath${\Phi}$}_n(s)ds\right\|_{\ell^1} 
& \leqs\int_0^\tau e^{-(1+\beta)(\tau - s)}\sum_{j=0}^n\sum_{k=0}^j\frac{(\tau-s)^k}{k!}\left|\Phi_{j-k}(0)\right|ds \nonumber \\
& = \int_0^\tau e^{-(1+\beta)(\tau - s)}\sum_{k=0}^n\frac{(\tau-s)^k}{k!}\sum_{p=0}^{n-k}\left|\Phi_{p}(0)\right|ds \nonumber \\
& \leqs \int_0^\tau e^{-(1+\beta)(\tau - s)}e^{\tau-s}\left\|\text{\boldmath${\Phi}$}(s)\right\|_{\ell^1} ds \nonumber \\
& =  \int_0^\tau e^{-\beta(\tau - s)}\left\|\text{\boldmath${\Phi}$}(s)\right\|_{\ell^1} ds, \label{bound2}
\end{align}
where $\text{\boldmath${\Phi}$} := (\Phi_i)_{i\in\Nb_0}^{\text{\sf T}}.$ To estimate $\left\|\text{\boldmath${\Phi}$}(s)\right\|_{\ell^1}$ observe that, because we have, for each $t\geqs 0$, $\widetilde{U} = (x, M_0, M_1, \ldots) \in X \subset \ell^1$, and
each component converges exponentially to the corresponding component of the limit equilibrium $\widetilde{U}^{\text{eq}}$,
and thus, in particular, $|x(\tau) - x^\text{eq}| \leqs C_1e^{-\eta \tau}$ for some $C_1, \eta >0,$ and all $\tau >0$, so that we get
\begin{equation}
\left|\Phi_0(s)\right| = \left| -\varphi_0(s)\right| = \frac{1}{x^\text{eq}}\left| M_0(s)\right| \,\left|x(s)-x^\text{eq}\right| 
\leqs \frac{C_1}{x^\text{eq}}e^{-\eta s}\left|M_0(s)\right|, \label{boundphi0}
\end{equation}
and, for $i\geqs 1,$
\begin{align}
\left|\Phi_i(s)\right| & =  \left| -\varphi_i(s) + \varphi_{i-1}\right| \,\leqs \, \left| \varphi_i(s) \right| + \left| \varphi_{i-1}\right| \nonumber \\
& \leqs   \frac{C_1}{x^\text{eq}}e^{-\eta s}\Bigl(\left|M_i(s)\right| + \left|M_{i-1}(s)\right| \Bigr). \label{boundphii}
\end{align}
Thus
\begin{align}
\left\|\text{\boldmath${\Phi}$}(s)\right\|_{\ell^1}  & = \sum_{i=0}^\infty \left|\Phi_i(s)\right| 
\leqs  \frac{2C_1}{x^\text{eq}}\bigl\|\widetilde{U} (s)\bigr\|_{\ell^1}e^{-\eta s}
\leqs  \frac{C_2}{x^\text{eq}}e^{-\eta s}, \label{boundPhi}
\end{align}
where $C_2 \geqs 2C_1 \max_{s\geqs 0}\Bigl\{\bigl\|\widetilde{U} (s)\bigr\|, \bigl\|\widetilde{U}^\text{eq}\bigr\|\Bigr\},$
and the maximum exists by the result about convergence in Theorem~\ref{strongstability}.
Hence, plugging \eqref{boundPhi} into \eqref{bound2}, we conclude that
\begin{align}
\left\|\int_0^\tau e^{J_{n+1}(\tau - s)}\text{\boldmath${\Phi}$}_n(s)ds\right\|_{\ell^1} 
& \leqs 
\begin{cases}
C_2\tau e^{-\beta \tau}, & \text{if $\eta=\beta$} \\  \frac{C_2}{|\beta-\eta|}e^{-\min\{\eta, \beta\}\tau}, &\text{if $\eta\neq\beta$,} \label{bound3}
\end{cases}
\end{align}
which, together with \eqref{bound1}, allow us to write, for all $\tau \geqs 0,$
\begin{equation}
\left\|\text{\boldmath${\delta}$}(\tau)\right\|_{\ell^1} \leqs 
 e^{-\beta\tau}\left\|\text{\boldmath${\delta}$}(0)\right\|_{\ell^1} + 
\begin{cases}
C_2\tau e^{-\beta \tau}, & \text{if $\eta=\beta$} \\  \frac{C_2}{|\beta-\eta|}e^{-\min\{\eta, \beta\}\tau}, &\text{if $\eta\neq\beta$.} \label{bound4}
\end{cases}
\end{equation}

\medskip

Let us now consider convergence in the norm of $X.$ Since
\begin{equation}
\|\text{\boldmath${\delta}$}(\tau)\| = \left|x(\tau)-x^\text{eq}\right| + \sum_{i=0}^\infty (i+1)\left|\delta_i(\tau)\right|
= \left|x(\tau)-x^\text{eq}\right| + \left\|\text{\boldmath${\delta}$}(\tau)\right\|_{\ell^1} + \left\|\text{\boldmath${\xi}$}(\tau)\right\|_{\ell^1},  \label{deltanormX}
\end{equation}
where $\text{\boldmath${\xi}$} = (\xi_i) :=  (i\delta_i).$ Multiplying \eqref{dotdeltai} by $i$ we get the system for $\xi_i$:
\begin{equation}\label{dotxii}
\xi'_i=-(1+\beta)\xi_i+\xi_{i-1} + \Psi_i(\tau),\qquad i=1,2,\dots\nonumber
\end{equation}
where $\Psi_i(\tau) :=\delta_{i-1}(\tau) + i\Phi_i(\tau),$ for $i\geqs 1$, and $\xi_0(\tau)\equiv 0.$
Again, like  \eqref{dotdeltai} this system can be solved recursively for $i=1, 2, \ldots,$ because the equation
for $\xi_i$ only depends on information with $j\leqs i,$ and so, similarly to what was done before, 
we can consider a finite $n$-dimensional for the vector $\text{\boldmath${\xi}$}_n = (\xi_1, \ldots, \xi_n)^{\text{\sf T}},$
\begin{equation}
\text{\boldmath${\xi}$}_n' = J_{n}\text{\boldmath${\xi}$}_n  + \text{\boldmath${\Psi}$}_n(\tau),\nonumber
\end{equation}
where $\text{\boldmath${\Psi}$}_n  = (\Psi_i)_{i=1,\ldots,n}^{\text{\sf T}}$, and
$J_{n}$ is the $n$-dimensional Jordan matrix with the form \eqref{Jmatrix}.
Now computations analogous to those done previously give the following decay estimate for 
$\left\|\text{\boldmath${\xi}$}(\tau)\right\|_{\ell^1}$ for $\tau$ sufficiently large: 
\begin{align}
\left\|\text{\boldmath${\xi}$}(\tau)\right\|_{\ell^1} &  \leqs
 e^{-\beta\tau}\left\|\text{\boldmath${\xi}$}(0)\right\|_{\ell^1} + 
\begin{cases}
C_3 \tau e^{-\beta \tau}, & \text{if $\eta>\beta$} \\  
C_4 \tau^2 e^{-\beta\tau}, &\text{if $\eta =\beta$.} \\  
C_5 e^{-\eta\tau}, &\text{if $\eta<\beta$,} \nonumber 
\end{cases}
\end{align}
where the constants $C_j$ are independent of $\tau$.
This, together with \eqref{bound4}, the exponential
decay bound for $ \left|x(\tau)-x^\text{eq}\right|,$ and \eqref{deltanormX}, allow us to conclude
that $\left\|\text{\boldmath${\delta}$}(\tau)\right\|$
converge exponentially fast to zero when $\tau\to +\infty$
which, recalling that $\tau=kx^\text{eq}t$,  proves the theorem.
\end{proof}

%
%

\section{Numerical explorations}\label{sec5}

In this section we present some of the numerical evidence illustrating the eigenvalues of 
the jacobian matrices $DF$ computed at the equilibria of \eqref{NonlinearODE}, 
for several values of the parameters $\alpha/r$ and $kr$, and for some dimensions $N+5$ of the system.

\medskip

The evidence presented illustrates properties described in Lemmas~\ref{lemmasimple} and \ref{imaginary}
and support Conjecture~\ref{conjonstability}. 

\medskip

The first evidence
consists in the plots of the numerical computed eigenvalues of the Jacobian $DF(U^*)$ of 
\eqref{NonlinearODE} at the bifurcation point $U^*$, when $\alpha/r = \mu^*$. We present in Figure~\ref{N10}
the spectra of this matrix for the system with $N=10$ (hence with dimension $N+5=15$) for several values of $kr$ from $0$ to 
$10^5$. The eigenvalues corresponding to small values of $kr$ are ploted in light gray and  cases with larger values of $kr$
become progressively darker. The spectra in the case of $kr=10^5$ is represented by the black dots. Note the existence of a
(black) point at the origin: this corresponds to the zero eigenvalue, whose existence
and simplicity, for all $kr$, was established in Lemma~\ref{lemmasimple}.

%
%
%
\begin{figure}[!h]
	\includegraphics[scale=0.85]{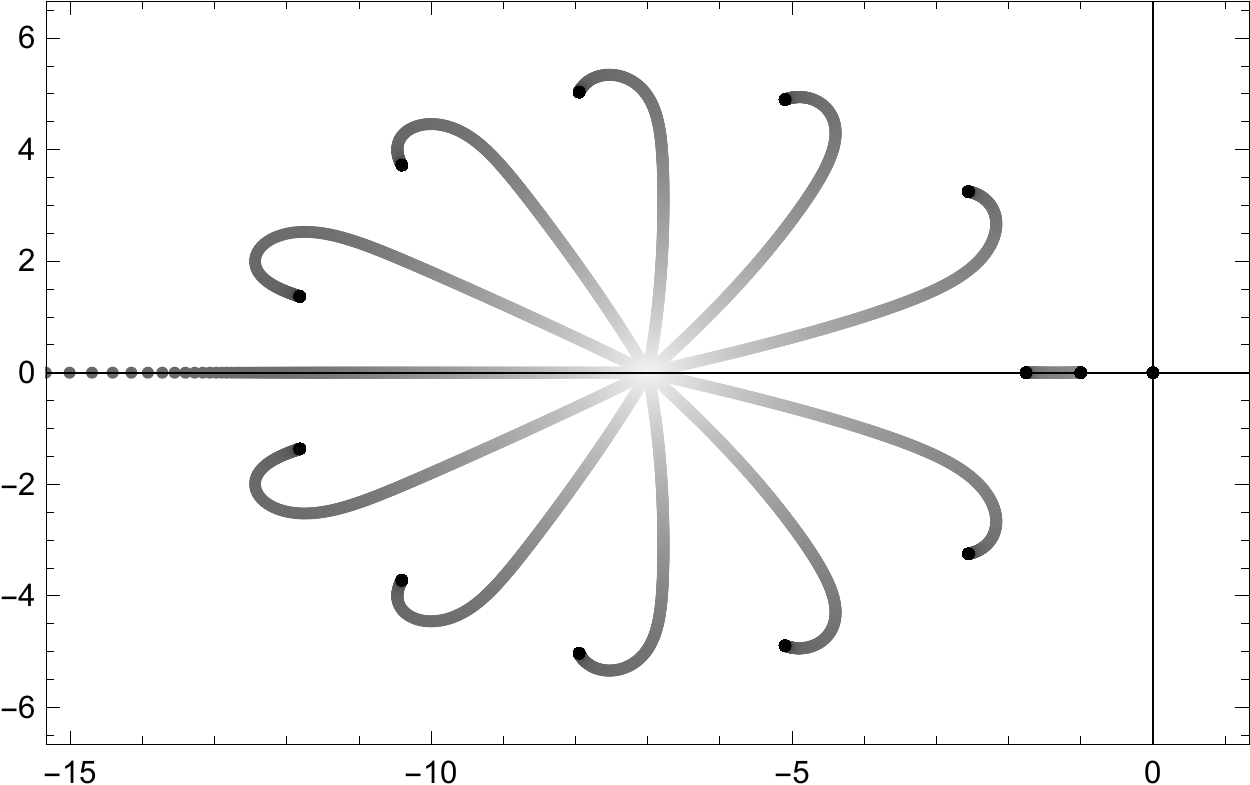}
	\caption{Plot of the eigenvalues of the jacobians $DF(U^*)$ with $N=10$, 
$kr$ from $0$ (light gray) to $10^5$ (black). The real eigenvalue with largest absolute value
gets out of the chosen window for $kr$ large enough.}\label{N10}
\end{figure}
%
%
%

\medskip

 In Figure~\ref{N25} the same plot is presented for the case $N=25$ and $kr$ from $0$ to 
$10$. In both cases it is clear that except for the 
zero eigenvalue, all other eigenvalues have negative real parts and
seem to remain bounded away from the imaginary axis when $kr$ increases.
Other experiments, for other values of $N,$ exhibit the same behaviour.


%
%
%
\begin{figure}[!h]\includegraphics[scale=0.85]{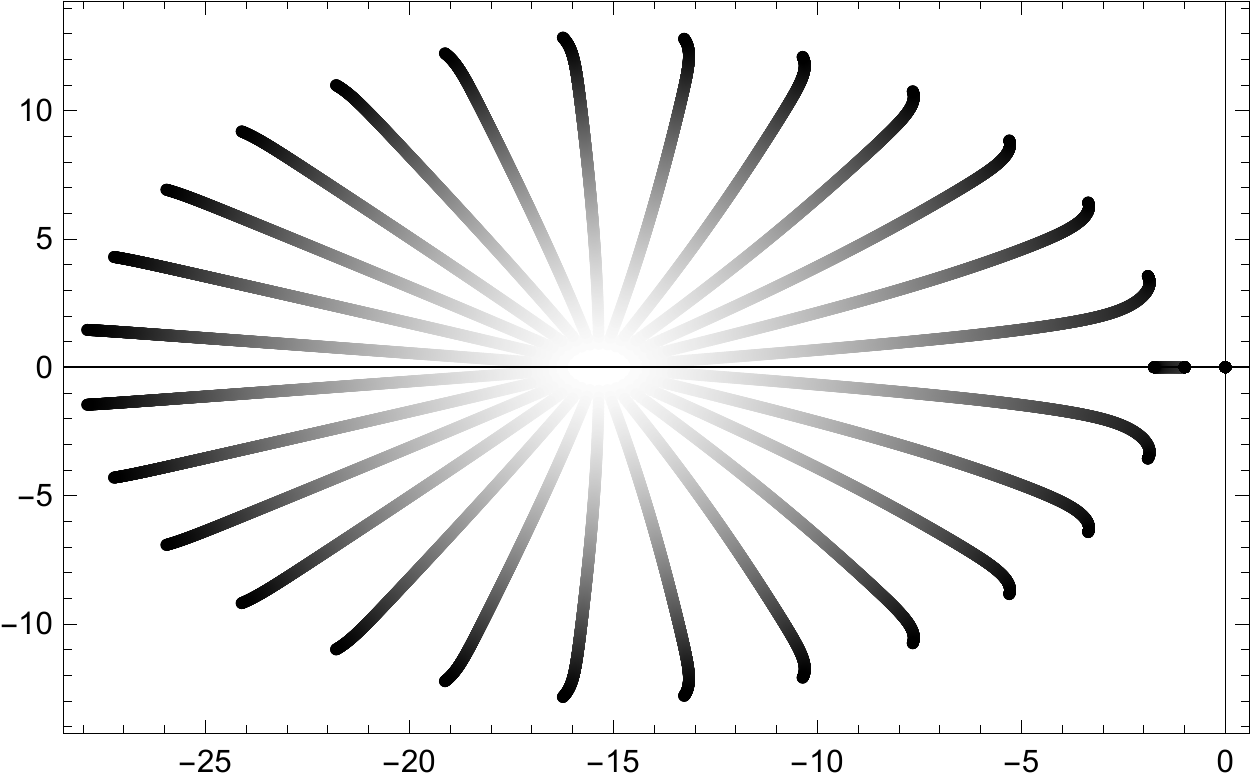}
	\caption{Plot of the eigenvalues of the jacobians $DF(U^*)$ with $N=25$, 
$kr$ from $0$ (light gray) to $10$ (black).}\label{N25}
\end{figure}
%
%
%

If $\alpha/r < \mu^*$ the corresponding experiments for the spectra of the jacobians $DF(U^{j*})$
about the two equilibria $U^{j*}$, with $j=1, 2$ (using
the notation of Theorem~\ref{stabilityatbifpoint}), shows a similar behaviour, except for the eigenvalue which was 
zero in the previous case (when $\alpha/r=\mu^*$) and is now negative for $j=1$ and positive for $j=2$. This is illustrated 
in Figures~\ref{N10_ymenor} and~\ref{N10_ymaior}. Observe
that in Figure~\ref{N10_ymaior} the eigenvalue that is zero when $kr=0$ becomes real positive when $kr>0$ but hardly moves 
at all. This behaviour is shown more clearly in Figure~\ref{estabilidade}.

%
%
%
\begin{figure}[!h]
	\includegraphics[scale=0.85]{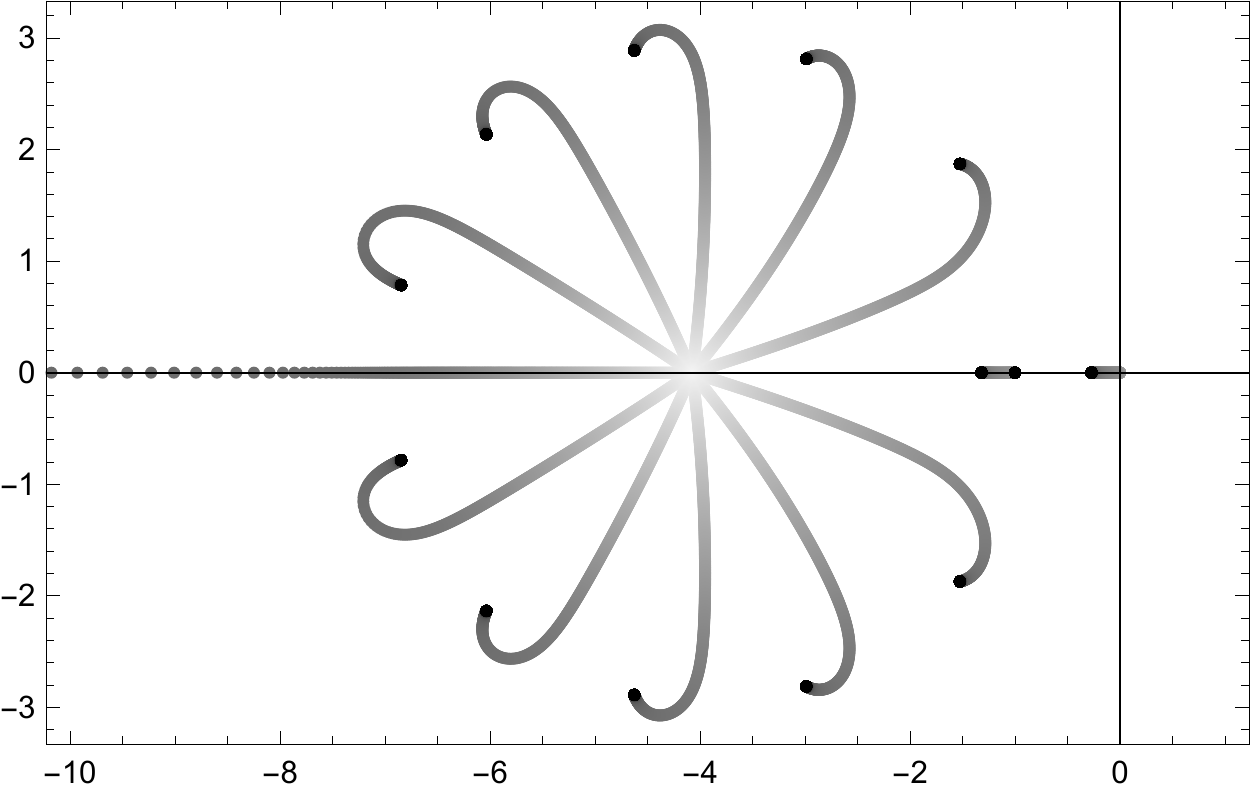}
	\caption{Plot of the eigenvalues of the jacobians $DF(U^{1*})$ with $N=10$, $\alpha/r =2.44$,  
$kr$ from $0$ (light gray) to $10^5$ (black). Observe the eigenvalue close to the origin starts at the origin when $kr=0$
and moves slowly to the left half plane as $kr$ increases. The real eigenvalue of largest absolute value
gets out of the chosen window for $kr$ large enough.}\label{N10_ymenor}
\end{figure}
%
%
%

%
%
%
\begin{figure}[!h]
	\includegraphics[scale=0.85]{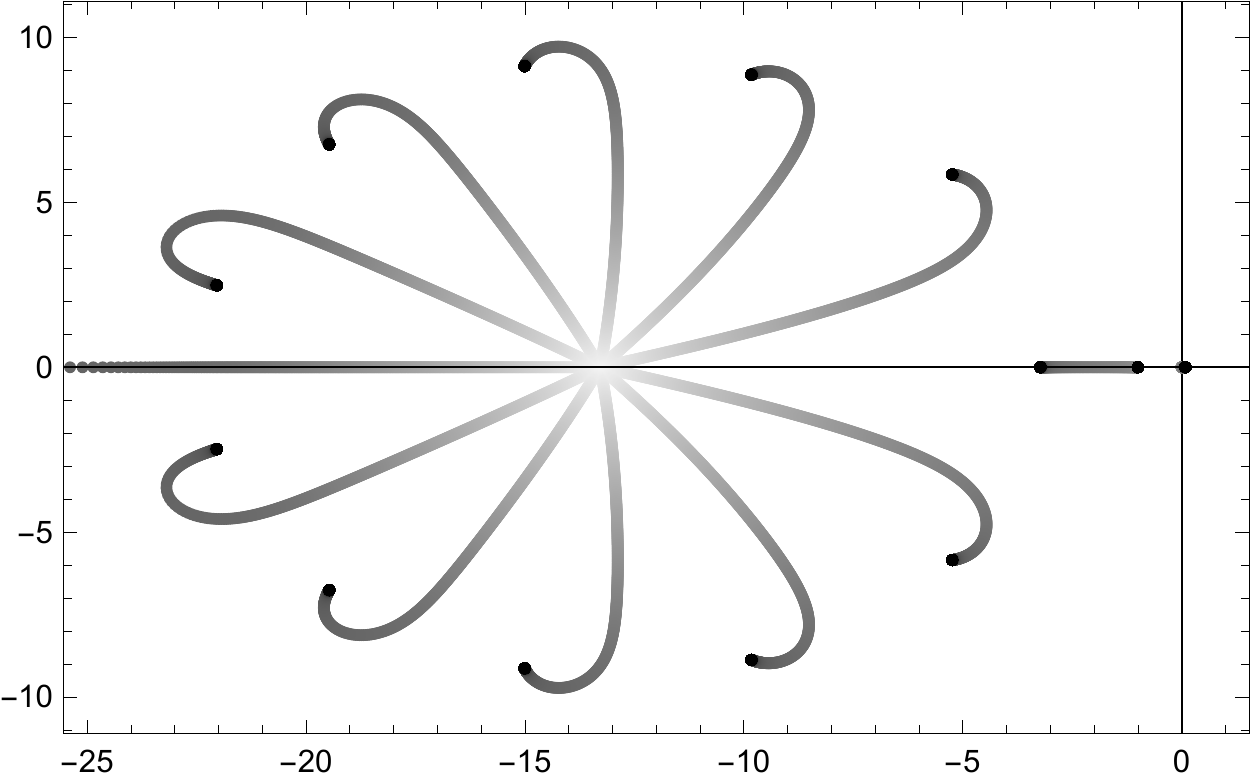}
	\caption{Plot of the eigenvalues of the jacobians $DF(U^{2*})$ with $N=10$, $\alpha/r = 2.44$,  
$kr$ from $0$ (light gray) to $10^5$ (black). Observe the eigenvalue close to the origin starts at the origin when $kr=0$
and moves slowly to the right half plane as $kr$ increases. The real eigenvalue of largest absolute value
gets out of the chosen window for $kr$ large enough.}\label{N10_ymaior}
\end{figure}
%
%
%

From all numerical experiments performed it is clear that for all $N$ and $kr$, except for the eigenvalue that is
equal to zero at the bifurcation value $\mu^*$, all other eigenvalues of the jacobians $DF(U^{j*})$ have negative 
real parts bounded away from zero.

\medskip

The final plot, in Figure~\ref{estabilidade}, plots, in a
window with $y^{j*}$ between 0.75 and 0.98, the values of the eigenvalue
of $DF(U^{j*})$ that is zero at the bifurcation point $(\mu^*,y^*)$ when $kr\in [0, 100].$
Superimposed to the graph we plot lines highlighting those eigenvalues for values of $y$ at the equilibria $U^{1*}$ and $U^{2*}$
for values of $\alpha/r$ equal to  $2.746$ (dotted line) and $2.436$ (dashed lines). The full line
is the value of $y^*$ of the critical equilibrium $U^*$, which corresponds to $\alpha/r = \mu^*\approx 2.881.$
Observe that the eigenvalues change very steeply from the zero eigenvalue when $kr$ is very close to $0$ but then they
remain essentially independent of $kr$ and never stray very far from the origin, as have already been
observed in Figures~\ref{N10_ymenor} and \ref{N10_ymaior}.

%
%
%
\begin{figure}[!h]
	\includegraphics[scale=0.65]{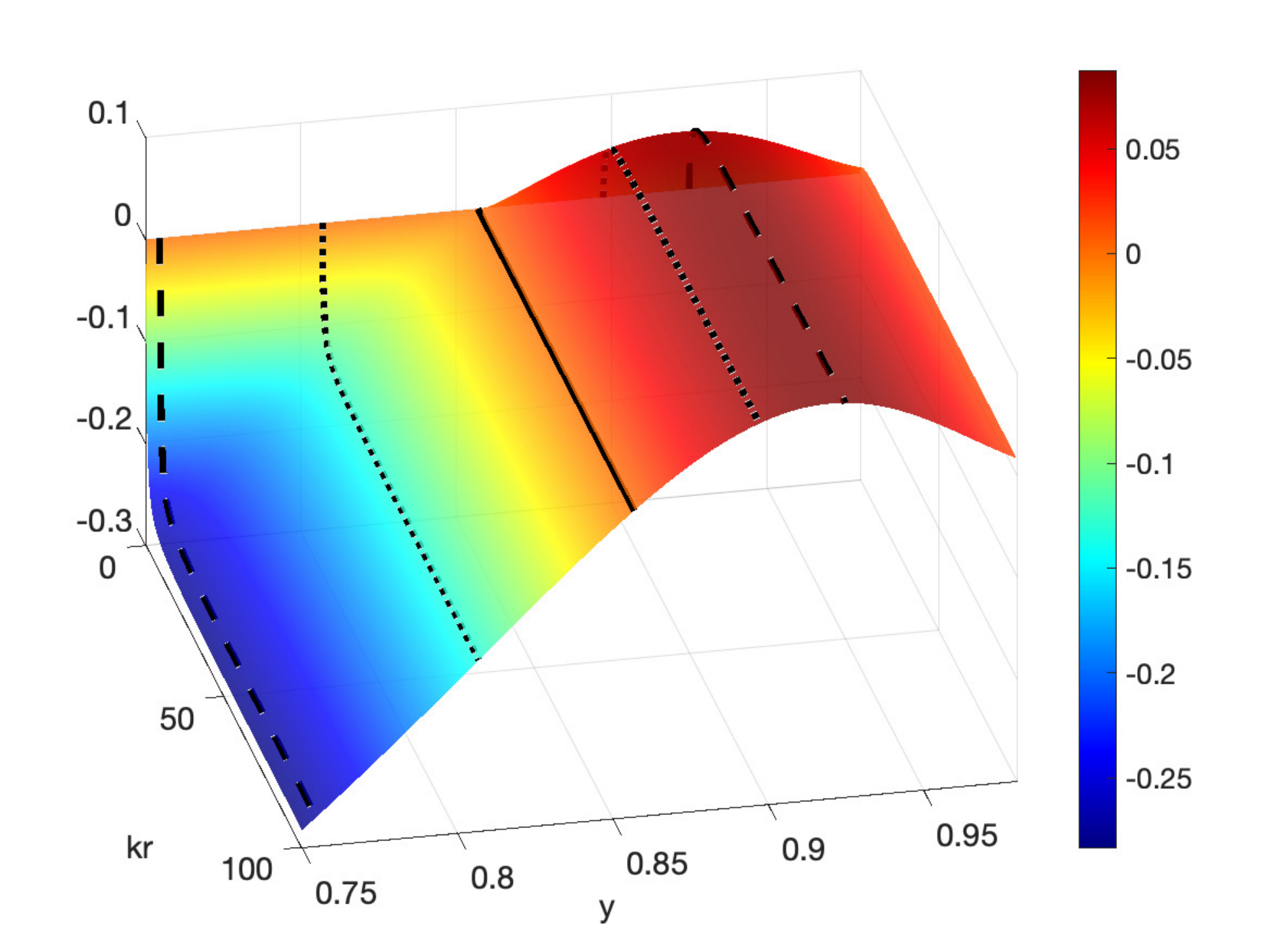}
	\caption{Plot of the eigenvalue of the jacobians $DF(U^{j*})$ with smaller absolute value,
when $N=10$,  $kr\in [0, 100],$ and $y$ in $[0.75, 0.98]$.
Superimposed to the graph we plot lines showing the values of $y$ at the equilibria $U^{1*}$ and $U^{2*}$
for values of $\alpha/r$ equal to  $2.746$ (dotted line) and $2.436$ (dashed lines). The full line
is the value of $y^*$ of the critical equilibrium $U^*$, which corresponds to $\alpha/r = \mu^*\approx 2.881.$
For any given value of $\alpha/r<\mu^*$ the line corresponding to $U^{1*}$ is always to the left of that of $U^{2*}.$}\label{estabilidade}
\end{figure}
%
%
%

%
%

\section{Discussion}\label{sec6}

In this paper we studied the local stability of equilibria of the model \eqref{syst2} for the silicosis disease, which is a
particular case of a more general model \eqref{M0eq1}--\eqref{xeq3} when the
special class of piecewise constant parameters \eqref{coef} is considered.

\medskip

With these assumptions it was known from \cite{cdg} that the balance between the input rates of silica and of new
macrophages, $\alpha$ and $r$ respectively, determined the existence (when $\alpha/r$ is below a certain threshold $\mu^*$)
or non-existence (when it is above) of equilibria of the infinite dimensional system \eqref{syst2}, as presented in the bifurcation 
diagram in Figure~\ref{figbif}.

\medskip

In this paper we proved that, for each $\alpha/r$ below the critical value $\mu^*$, the equilibrium with
smaller value of $x$ is a locally exponentially asymptotically stable solution of \eqref{syst2} in the strong
topology of the space $X\subset\ell^1$ of sequences with finite number of particles per unit volume introduced
in \cite{cps}. We prove also that the equilibrium solutions with larger value of $x$ are unstable.

\medskip

This stability result is proved by considering an appropriate change of variables \eqref{uvw} that allows us to write
\eqref{syst2} in the form \eqref{syst3} in which a closed finite dimensional subsystem can be identified. The analysis of the
eigenvalues of the linearizations of this finite dimensional system about the equilibria is the
crucial step to conclude the stability results for the original infinite dimensional model.

\medskip

To biologically interpret this result we observe that, having a constant input rate $\alpha$ of silica particles into the system, the
only way the system can converge to a non-negative 
steady state (with a finite concentration $x^{\text{eq}}$ of silica particles) is
if the mechanism eliminating silica particles by transporting them  inside 
the macrophages through the mucociliary escalator off the respiratory system
is highly efficient. From the results in this paper, this can only occur in this model if both the following conditions hold: (i) 
the rate of input of macrophages $r$ is sufficiently large compared with the input
rate of silica $\alpha$ (so that $\alpha/r$ is below the threshold $\mu^*$), and (ii) the initial load of silica in the
system is sufficiently small, so that the initial condition is inside the attraction basin of the asymptotically stable equilibrium.
If at least one of these conditions fails to hold, then solutions to \eqref{syst2} do not converge to an equilibrium (which
do not even exist if (i) fails). The rigorous
study of what happens in those cases is still lacking. However, preliminary
numerical studies (not presented in this paper) suggest that,
in those cases, solutions are such that $x(t)$ increase without bound. This unbounded increase in the amount of
silica dust in the respiratory system is the way this model expresses the fatal run off of the amount of crystalline quartz dust
in the lungs leading to death.

\medskip

It is an interesting mathematical open problem to study this run off regime and to investigate if it corresponds to
some self-similar regime, as is the case in other types of coagulation equations with inputs \cite{crw,cs,ffv}.

\medskip

Other mathematically interesting open problems arise by considering 
systems \eqref{M0eq1}--\eqref{xeq3} with more general rate coefficients $k_i, p_i$ and $q_i$,
in particular those satisfying power laws in the variable $i$ considered in \cite{cdg}. The study of those systems
will require a more precise
enquiry into the exact number of equilibria than was achieved in \cite{cdg} and, likely, a different way to attack
the stability problem in the infinite dimensional system \eqref{M0eq1}--\eqref{xeq3}, 
as the trick of using a change of variables to decouple
the system into a closed finite dimensional subsystem determining the dynamics is unlikely to be applicable in the general case.
However, based on the results about the structure of equilibria proved in \cite{cdg}, 
we expect the results in this paper 
to extend to systems with more general coefficients satisfying  power law assumptions.

%
%
\bibliographystyle{amsplain}

\end{document}